\documentclass[lettersize,journal]{IEEEtran}

\usepackage[T1]{fontenc}
\usepackage{cite}
\usepackage{float}
\usepackage{amsmath,amssymb,amsfonts,amsthm}
\usepackage{mathtools}
\usepackage{algorithm}
\usepackage{algorithmic}
\usepackage{graphicx}
\usepackage{textcomp}
\usepackage{xcolor}
\usepackage{booktabs}
\usepackage{multirow}
\usepackage{array}
\usepackage{colortbl}
\usepackage[hyphens]{url}
\usepackage[colorlinks=true,linkcolor=blue,citecolor=blue,urlcolor=blue]{hyperref}
\usepackage{microtype}
\usepackage{xspace}
\usepackage{tikz}
\usetikzlibrary{positioning,arrows.meta,calc,shapes.geometric,fit,backgrounds}

\newtheorem{theorem}{Theorem}
\newtheorem{lemma}{Lemma}
\newtheorem{proposition}{Proposition}

\newtheorem{assumption}{Assumption}

\newcommand{\halomark}{HaloMark\xspace}
\newcommand{\Wmat}{W}
\newcommand{\Sigmacalib}{\Sigma_{\mathrm{calib}}}
\newcommand{\effrank}{\mathrm{eff\_rank}}

\begin{document}

\title{\halomark: A Spectral Threshold for Embedding-Vector Watermarking under C2PA}

\author{Tarun~Sharma%
\thanks{T. Sharma is an independent researcher
(e-mail: tarun.sharma@ieee.org).}%
\thanks{Reproducible code and data:
\url{https://github.com/tarun-ks/halomark}.}}

\maketitle

\begin{abstract}
Foundation-model embeddings are now a primary data asset, but the
content-provenance machinery built for images and audio does not
transfer to them. C2PA binds to an asset with a stable bit-level or
perceptual identity; embeddings have neither, since quantisation,
projection, fine-tuning, and windowed averaging reshape them in normal
use and break any fixed hash.

We present \halomark{}, a watermark for embedding vectors
cryptographically bound to a C2PA manifest. It composes four standard
primitives---a block-diagonal orthogonal rotation, public whitening, an
input-dependent LSH commitment, and a per-vector nonce---around one
protocol change: the producer signs the LSH commitment $c$ into the
C2PA sidecar, and the verifier reads $c$ from the manifest instead of
recomputing it. Recomputing is fragile under whitening, which flips the
commitment bucket on $62\%$ of inputs at $\cos = 0.96$; reading the
signed $c$ reduces the verifier's score to
$T = T_{\mathrm{null}} + \beta(\mathcal{A})\cdot\varepsilon$, so
security turns on a single scalar $\beta$, which we bound rigorously for
linear and non-adaptive attackers and characterise empirically for the
adaptive case.

We evaluate against an adversary holding polynomially many
clean/watermarked pairs under one key with full sidecar visibility,
across eight baselines and ten adaptive attackers including
denoising-autoencoder removal. The eleven encoders separate at an
empirical threshold $\effrank(\Sigma)/d \approx 0.19$: above it,
detection AUROC stays at $0.98$ or higher across every in-budget attack
on the three encoders we sweep in full, and at $0.965$ or higher under
single-seed DAE removal on the rest; below it every variant we tested
fails. Why the threshold is dimension-uniform is left open. Deployed as
a Qdrant admission filter, the verifier runs at $284\,\mu$s and $24$
bytes of sidecar per vector, validated end-to-end against three C2PA
reference-SDK bindings.
\end{abstract}


\section{Introduction}\label{sec:introduction}

Embeddings have become a primary data asset. Retrieval-augmented
generation, semantic search, recommendation, biometric identification,
and content provenance all depend on dense vectors produced by
foundation models, and those vectors are now treated as content in
their own right: uploaded to vector databases, shared across
organisational boundaries, sold as datasets, and ingested by
downstream models.

C2PA-style provenance binding does not extend to dense vectors.
Routine quantisation, projection, and fine-tuning reshape every bit,
so a cryptographic hash on the embedding rarely survives legitimate
use. We want a watermark on the embedding itself that
survives those routine transformations, leaves cosine-based
downstream tasks intact, is cryptographically bound to a C2PA
manifest, and verifies at microsecond scale. Existing embedding
watermarks address some of these properties individually but none
together, and none under a threat model that includes a
denoising-autoencoder attacker with polynomial $(x, x')$ training
pairs.

\subsection*{The protocol idea: publish the commit}

A key-derived embedding watermark wants a content-dependent
signature in which the per-block perturbation $s_i$ depends on a
commitment $c\!=\!\texttt{Commit}(K, x; \mathrm{nonce})$ over the
input. SEAL~\cite{arabi2025seal} uses this pattern for image
watermarks, with a verifier that recomputes $c$ from the perturbed
test input and re-derives the expected signatures.

\textbf{Recomputing $c$ at verification time is fragile under
whitening.} Whitening
$\Wmat\!=\!(\Sigmacalib\!+\!\lambda_{\mathrm{reg}} I)^{-1/2}$
amplifies perturbations along low-eigenvalue directions by up to
$1/\sqrt{\lambda_{\mathrm{reg}}}\!\approx\!100\times$. The
recompute-commit verifier therefore loses roughly $20$ percentage
points of TPR on small-perturbation attacks even when the watermark
signal survives. \S\ref{sec:dirOrac} reports the full bucket-flip
rate by attack class.

Our fix is to publish $c$ in the C2PA-signed sidecar. The producer
outputs $(x', n, c)$, the manifest signs all three fields, and the
verifier reads $c$ directly without recomputing it from $\hat{x}$.
Tampering with $c$, $n$, or $x'$ breaks the C2PA signature before
\texttt{Verify} runs. Figure~\ref{fig:overview} shows the protocol
layout.

\begin{figure*}[t]
\centering
\begin{tikzpicture}[
  >=Stealth,
  font=\small,
  every node/.style={font=\small},
  stage/.style={draw, rounded corners=2pt, minimum height=9mm,
               minimum width=22mm, align=center, inner sep=3pt,
               line width=0.5pt},
  prod/.style={stage, fill=blue!8,  draw=blue!50!black},
  ver/.style={stage,  fill=green!8, draw=green!40!black},
  side/.style={draw, rounded corners=3pt, fill=orange!15,
               draw=orange!70!black, line width=0.8pt,
               align=center, inner sep=5pt,
               minimum height=14mm, minimum width=48mm},
  arr/.style={->, semithick, >=Stealth, draw=black!70},
  publishedge/.style={->, line width=1.4pt, >=Stealth,
                       draw=red!75!black,
                       dash pattern=on 4pt off 2pt},
  pubdash/.style={->, dashed, line width=0.5pt, >=Stealth,
                   draw=orange!70!black},
  rowlabel/.style={font=\small\itshape, anchor=east, text=black!70},
]

\node[rowlabel] at (-0.4, 3.2)  {Producer};
\node[prod] (X)   at ( 1.0, 3.2) {clean $x$};
\node[prod] (PW)  at ( 4.0, 3.2) {whiten\\$W = \Sigma^{-1/2}$};
\node[prod] (PQ)  at ( 7.0, 3.2) {rotate $Q$\\(block-diag)};
\node[prod] (LSH) at (10.0, 3.2) {LSH-commit\\$c$};
\node[prod] (PRF) at (13.0, 3.2) {sign\\$\{s_i\} = \mathrm{PRF}(K, n, c)$};
\node[prod] (Add) at (16.0, 3.2) {add $\varepsilon\,s_i$\\$\to x'$};

\draw[arr] (X)   -- (PW);
\draw[arr] (PW)  -- (PQ);
\draw[arr] (PQ)  -- (LSH);
\draw[arr] (LSH) -- (PRF);
\draw[arr] (PRF) -- (Add);
\draw[arr] (PQ.south) to[bend right=18] (Add.south);

\node[side] (sidecar) at (8.5, 1.0)
   {\textbf{C2PA-signed sidecar}\\[1pt]
    $\bigl(\, x',\;\; n,\;\; \boxed{c}\,\bigr)$};

\draw[pubdash] (Add.south)  to[bend left=10]  (sidecar.east);
\draw[pubdash] (LSH.south)  -- node[right=1pt, font=\scriptsize, text=orange!70!black] {publish $c$}
                              (sidecar.north);

\node[rowlabel] at (-0.4, -1.4) {Verifier};
\node[ver] (Xh)   at ( 1.0, -1.4) {test $\hat x$};
\node[ver] (VW)   at ( 4.0, -1.4) {whiten\\$W$};
\node[ver] (VQ)   at ( 7.0, -1.4) {rotate $Q$};
\node[ver] (Der)  at (10.0, -1.4) {derive\\$\{s_i\} = \mathrm{PRF}(K, n, c)$};
\node[ver] (Sc)   at (13.0, -1.4) {score\\$T = \langle Q\hat x,\, s\rangle$};
\node[ver] (Tt)   at (16.0, -1.4) {$T \geq \tau$?};

\draw[arr] (Xh)  -- (VW);
\draw[arr] (VW)  -- (VQ);
\draw[arr] (VQ)  -- (Der);
\draw[arr] (Der) -- (Sc);
\draw[arr] (Sc)  -- (Tt);

\draw[publishedge] (sidecar.south)
   to[bend left=8]
   node[right=2pt, align=left, font=\small\bfseries,
        text=red!70!black, pos=0.55]
        {read $c$\\(no recompute)}
   (Der.north);

\end{tikzpicture}
\caption{Published-commit verifier protocol. The producer's LSH commit
$c$ travels with $(x', n)$ in the C2PA sidecar (centre). The verifier
reads $c$ from the manifest to derive $\{s_i\}$ instead of recomputing
$\texttt{Commit}(K, \hat{x})$. The red edge is the only protocol change
vs.\ a recompute-commit verifier; everything else is standard
C2PA-bound watermark verification. \S\ref{sec:dirOrac} measures the
fragility this change removes.}
\label{fig:overview}
\end{figure*}
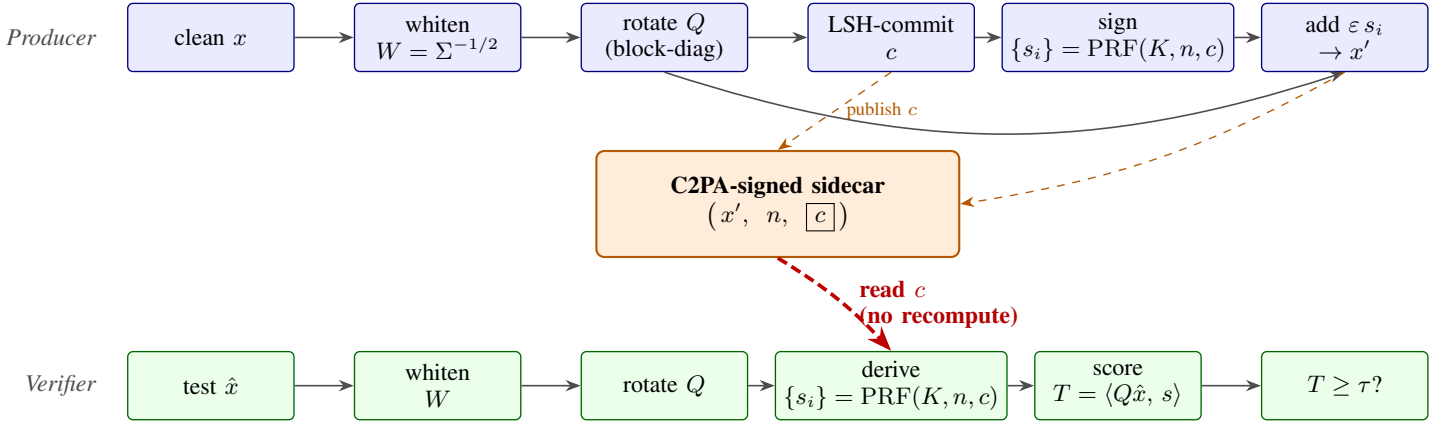

\subsection*{Contributions}

The headline finding is empirical. Across 11 embedders and six
structurally distinct perturbative approaches we find an empirical
spectral threshold (\S\ref{sec:threshold}) near
$\effrank(\Sigma)/d \approx 0.19$: above it the construction defends
at AUROC $0.96$ or higher; below, every approach collapses to a
manifold-projection attacker trained on as few as 10 KPA pairs. The
threshold is dimension-uniform but not yet derived; why is left open.
Whether a perturbative embedding watermark exists at all on a given
encoder is decided by the encoder's spectrum, not by the watermark
design. We adopt the manifold-aligned variant as the default,
gaining 10 to 19 percentage points of $\beta$ over a random-$Q$
baseline on above-threshold encoders.

Three supporting contributions make the upper half of that story
attainable. First, the published-commit verifier protocol
(\S\ref{sec:core}) reduces the security argument to a single
signature-retention scalar via the linear decomposition of
Eq.~\eqref{eq:score-decomp} (Lemma~\ref{thm:retention}). Second, a closed-form Wiener bound
(Theorem~\ref{thm:wiener-floor}) characterises the optimal
adaptive-linear attacker through a Gaussian first-order term
$\beta_{\mathrm{char}}\!\approx\!\phi(\nu^*)/\sqrt{d\,\rho(\nu^*)}$
within about $13\%$ of the empirical Wiener attack, plus a
conservative Cauchy--Schwarz lower bound ($\beta\!\geq\!0.136$ on
MiniLM, well below the empirical floor of $0.851$ but
distribution-free); both are computable in $O(d)$ from the
encoder's eigenspectrum. Third, Mode C (\S\ref{sec:mode-c}) adds a
text-MinHash sketch to the C2PA manifest, rejecting $92\%$ of
multi-attempt paraphrase attacks while accepting $100\%$ of routine
perturbations, and so closes the paraphrase regime in RAG-style
deployments where source text accompanies the embedding.

Deployment-scale results (\S\ref{sec:baselines}): cross-corpus
validation, 64-key multi-tenant isolation up to a 32-key collusion
DAE, self-calibration (Mode B), 4-bit-rep4 multi-bit payload
decoding, Qdrant 1.17 admission-filter integration at $284\,\mu$s
per vector, and end-to-end interop with three independent language
bindings of the C2PA reference SDK.

\paragraph*{Threat model}
The C4 attacker holds polynomial-many $(x, x', n, c)$ tuples under
one key plus full sidecar visibility, with no verifier-oracle access.
We exercise a 14-row standard attack sweep together with a series of
strong DAE variants: training on up to $10^6$ pairs, $4\times$
width, four residual layers, FiLM conditioning, a $65{,}536$-bucket
gating head, spectrally normalised denoisers with total Lipschitz
constant in $[1.5, 300]$, natively budget-constrained DAEs, and a
within-C4 surrogate-PGD attacker. Above the threshold, AUROC stays
at $0.99$ or better across every in-budget attacker we tried; below it even a
10-pair DAE defeats the watermark.

\section{Background and Threat Model}\label{sec:threat}

\subsection{Embeddings as protected content}\label{sec:embeddings-as-content}
Dense neural embeddings are vectors in $\mathbb{R}^d$, typically with $d \in [384, 4096]$, produced by a foundation model on input content. The relevant invariants are: cosine similarity between related content should be high; cosine similarity between unrelated content should be low; the specific coordinate values are not semantically meaningful, only their geometric relationships are. This last property --- \emph{rotation invariance of utility} --- is exactly what \halomark{} exploits: we can rotate the embedding into a secret basis, perturb selected coordinates, and rotate back, without disturbing any downstream task that depends only on cosine similarity between embeddings.

In a deployed RAG or semantic-search system, embeddings are produced once and then ingested into a vector database, where they are queried by cosine similarity. The vector is the asset; the manifest must travel with it for provenance to remain attached. \halomark{} provides exactly this: a watermark embedded into the vector itself, cryptographically bound to a C2PA manifest such that neither the vector nor the manifest can be independently substituted without breaking verification.

\subsection{C2PA manifests}\label{sec:c2pa-bg}
The C2PA specification~\cite{c2pa2024} defines a signed manifest containing a sequence of assertions about a piece of content --- origin, processing history, cryptographic hashes, and author signatures. Manifests are designed for media. Two properties matter for \halomark: (a) the manifest contains a deterministic claim hash that can be hashed into a key, and (b) the manifest signature allows a verifier to authenticate manifest content without trusting the transport. \halomark{} treats the manifest as an external oracle and does not modify the C2PA specification.

\subsection{Threat model: Multi-Party Adversary (MPA)}\label{sec:mpa}
We consider a pipeline with three parties: an \emph{embedding producer}~$E$ (typically a model provider) who generates watermarked embeddings; a \emph{vector-store operator}~$V$ who indexes and serves embeddings; and a \emph{consumer}~$C$ who retrieves embeddings and uses them downstream. \halomark's threat model assumes that any one of these parties may be adversarial, but the C2PA-signed manifest's signing key is not compromised.

Formally, an MPA adversary~$\mathcal{A}$ has the following capabilities:

\begin{itemize}
\item \textbf{C1 (Observation).} $\mathcal{A}$ sees any watermarked embedding $x'$ and the corresponding C2PA manifest~$m$.

\item \textbf{C2 (Routine post-processing).} $\mathcal{A}$ may apply dimensionality reduction (PCA, random projection), quantisation (int8, int4, 1-bit), normalisation, clipping, and convex combination (averaging).

\item \textbf{C3 (Fine-tuning drift).} $\mathcal{A}$ may ingest $x'$ into a downstream model and produce a perturbed $x''$ after gradient updates.

\item \textbf{C4 (Targeted removal under known-plaintext access \emph{and} a semantic-preservation budget).}
$\mathcal{A}$ computes $\tilde{x}$ such that $\cos(\tilde{x}, x_{\mathrm{clean}}) \geq \delta_{\mathrm{C4}}$ \emph{and} $\mathrm{Verify}(K, \tilde{x}, \tau_{\mathrm{det}})$ rejects, where $x_{\mathrm{clean}}$ is the original unwatermarked embedding and $\delta_{\mathrm{C4}}$ is a deployment-driven cosine threshold (canonical operational minimum: $\delta_{\mathrm{C4}}{=}0.95$).

\emph{Why $\delta_{\mathrm{C4}}{=}0.95$.}
$\delta_{\mathrm{C4}}$ is a single deployer-chosen value; attacks past
whichever value is chosen count as destroying the asset for that
deployment. $\delta_{\mathrm{C4}}{=}0.95$ is conservative for
biometric, deduplication, and near-duplicate detection
pipelines~\cite{karpukhin2020dpr,reimers2019sbert}; broad ranked
retrieval (MS~MARCO MRR@10) admits a looser
$\delta_{\mathrm{C4}}{=}0.85$ (the MRR@10 deltas at
$\cos\!\in\!\{0.95, 0.85, 0.80\}$ are in \S\ref{sec:imperc}). We
report headline numbers at $\delta_{\mathrm{C4}}{=}0.95$ and mark
rows admitted only at the relaxed $0.85$ budget (Table~\ref{tab:rob},
$\ddagger$). The L2-distance equivalent on unit-normalised embeddings
is $\sqrt{2(1-\delta_{\mathrm{C4}})}$.

\emph{Adversary access.}
$\mathcal{A}$ has access to polynomial-many tuples $(x_i, x'_i, \mathrm{nonce}_i, c_i)$ under the same key --- arbitrarily many clean/watermarked pairs \emph{together with} the producer-published per-vector nonce and commitment from the manifest sidecar (\S\ref{sec:core}) --- and full knowledge of the \halomark{} algorithm. $\mathcal{A}$ does \emph{not} have access to the key~$K$ itself, nor to a \texttt{Verify} oracle (excluded in this section, reinforced in \S\ref{sec:disc}). The published $(\mathrm{nonce}_i, c_i)$ are \emph{strictly more} information than a recompute-commit construction would give $\mathcal{A}$; we still defend (\S\ref{sec:security}) under this stronger model.

\emph{Outcome accounting.}
Attacks that achieve \texttt{Verify} rejection by perturbing $x'$ past the budget ($\cos(\tilde{x}, x_{\mathrm{clean}}) < \delta_{\mathrm{C4}}$) destroy the asset rather than strip the watermark and so are not counted as successful removal --- utility loss is itself the deployment-side defence. C4 directly mirrors the known-plaintext-style forgery and removal attacks shown by M\"{u}ller~et~al.~\cite{muller2025blackbox} and follow-up work~\cite{dong2025wmcopier,zhu2025pnp} to break content-agnostic image and diffusion-model watermarks, extended to embeddings and to the per-output-randomness setting of Bileve~\cite{zhou2024bileve}. \S\ref{sec:dirOrac} reproduces the analogous attack against a content-agnostic ablation of \halomark{} and motivates the full construction of \S\ref{sec:core}.

\item \textbf{C5 (Forgery).} $\mathcal{A}$ attempts to produce a fresh embedding $z$ and manifest $m'$ such that \halomark{} verification accepts $(z, m')$.

\item \textbf{C6 (Collusion).} $\mathcal{A}$ observes many $(x', m)$ pairs and attempts to extract the master key or block structure.
\end{itemize}

We formalise the game-based definitions in Appendix~\mbox{S-I}. Outside the MPA model:
\begin{itemize}
\item We do not address attacks that compromise the C2PA signing key itself.
\item We do not address adversaries who control the embedding model's training data at origin.
\item We assume the calibration-corpus registry (which serves $\Sigmacalib$ keyed by \texttt{calibration\_corpus\_id}) is trusted at the same level as a C2PA-CA. A compromised registry could publish a low-rank backdoored $\Sigma$ that aligns the watermark with attacker-knowable directions; we recommend either pinning the SHA-256 of $\Sigma$ in the producer's manifest or recomputing $\Sigma$ locally from a publicly-available corpus. We discuss this in \S\ref{sec:c2pa} and \S\ref{sec:disc}.
\item We do not address adaptive adversaries with access to a \texttt{Verify} oracle under the target key. Such an adversary can in principle learn the watermark subspace through staircase attacks; \halomark{} assumes the verifier exposes \emph{no} binary-decision interface to untrusted parties. We discuss deployment hardening in \S\ref{sec:disc}.
\item We do not defend against passive side-channel leakage (timing, cache) from the detection oracle.
\end{itemize}

These are orthogonal security properties better handled by C2PA, by standard key-management, and by deployment-side hardening, respectively.

\subsection{Security goals}\label{sec:goals}
\halomark{} targets three properties, stated informally:

\begin{itemize}
\item \textbf{Imperceptibility.} A watermarked embedding $x'$ is indistinguishable from an unwatermarked $x$ from the perspective of any downstream cosine-similarity-based task, up to an adjustable parameter $\varepsilon$ (\S\ref{sec:imperc}).
\item \textbf{Robustness.} For all $\mathcal{A}$-attacks in C2--C4, detection AUROC $\geq \tau_{\mathrm{robust}}$ with $\tau_{\mathrm{robust}} \geq 0.85$ on every C4-valid operating point (\S\ref{sec:rob}).
\item \textbf{Unforgeability.} For all $\mathcal{A}$ attempts at C5, verification success probability is at most $2^{-\lambda}$ for security parameter $\lambda$ (formalised as Thm.~S1, App.~\mbox{S-I}).
\end{itemize}

Formal definitions appear in Appendix~\mbox{S-I}.

\section{Related Work}\label{sec:related}

\subsection{Positioning relative to prior watermarks}\label{sec:related-table}
\halomark{} differs from prior watermarks on four axes: KPA-resist
(defends polynomial $(x, x')$-pair attackers), per-vector freshness,
content-dependence, and C2PA-binding. Existing EaaS embedding
watermarks (EmbMarker~\cite{peng2023embmarker},
WARDEN~\cite{shetty2024warden}, WET~\cite{shetty2025wet},
ESpeW~\cite{wang2025espew}) have none of these (they target paraphrase/
clustering removal, not KPA). SEAL~\cite{arabi2025seal} adds
content-dependence for images but recomputes the commit at verify time;
Bileve~\cite{zhou2024bileve} adds per-token freshness for LMs;
MetaSeal~\cite{zhou2025metaseal} adds C2PA-binding for images. Only
\halomark{} satisfies all four for embedding vectors --- the
published-commit step (\S\ref{sec:core}) is what makes the all-yes row
tractable under whitening on retrieval encoders.


\subsection{Embedding watermarking}\label{sec:related-emb}
Foundational deep-NN watermarking starts with
Uchida~et~al.~\cite{uchida2017embedding} (in-weight payloads) and
Adi~et~al.~\cite{adi2018turning} (backdoor-based ownership claims),
systematised in Lukas~et~al.~\cite{lukas2022sok}.
Saberi~et~al.~\cite{saberi2024robustness} raise the bar for what counts
as a strong adversary. The closest prior art is the
\emph{Embeddings-as-a-Service (EaaS) watermarking} line, which protects
copyright over dense embeddings against imitation and removal:
\textbf{EmbMarker}~\cite{peng2023embmarker} (trigger-based, content-agnostic
direction); \textbf{WARDEN}~\cite{shetty2024warden} (multi-orthogonal
spread for clustering-removal resistance); \textbf{ESpeW}~\cite{wang2025espew}
(sparse content-dependent mask); \textbf{WET}~\cite{shetty2025wet}
(linear-transformation pattern). \emph{None target a polynomial
known-plaintext adversary.} Two 2025 semantic EaaS watermarks postdate
this design: \textbf{RegionMarker}~\cite{yang2025regionmarker} triggers
watermarks on PCA-defined regions, and \textbf{SemMark}~\cite{li2025semmark}
uses LSH to partition the semantic space and inject region-local
signatures; both defend model-extraction copyright rather than
per-vector provenance, and neither admits a polynomial KPA adversary.
A concurrent geometry-aware localized scheme~\cite{chen2026geomark}
likewise places the mark by data geometry but stays within the same
EaaS-copyright model.
SemMark's LSH partitioning is the closest prior use of locality-sensitive
hashing, but it hashes to \emph{place} the watermark, whereas we publish
the commitment $c$ in the C2PA sidecar so the verifier never recomputes
it. Our \texttt{WET-approx} re-implementation
collapses to AUROC $\approx 0.50$ under direction-oracle (\S\ref{sec:baselines}),
empirically confirming the structural argument. Faithful ports of
EmbMarker and WARDEN reduce cell-by-cell to \texttt{AddDir} and
\texttt{WET-approx} on our threat model (Table~\ref{tab:baselines}).

Adjacent constructions reused here:
\textbf{SEAL}~\cite{arabi2025seal}'s LSH-commit pattern (originally for
images); \textbf{Bileve}~\cite{zhou2024bileve}'s per-output randomness
for LM token watermarks; \textbf{MetaSeal}~\cite{zhou2025metaseal}'s
C2PA-binding pattern. The \emph{published-commit} step (\S\ref{sec:core})
is the new piece: SEAL recomputes $c$ at the verifier, which is
fragile under whitening (\S\ref{sec:dirOrac}); we publish $c$ in the
C2PA sidecar so the verifier never recomputes it. Token-stream text watermarks
(\textbf{KGW}~\cite{kirchenbauer2023kgw},
Aaronson~\cite{aaronson2022scott},
Christ~et~al.~\cite{christ2024undetectable}) and pixel-space image
watermarks (HiDDeN~\cite{zhu2018hidden}, RoSteALS~\cite{bui2023rosteals})
live in different signal domains; we cite Tree-Ring~\cite{wen2023treering}
and Stable Signature~\cite{fernandez2023stablesig} as the
diffusion-output reference. We measure mutual non-interference with
KGW on $500$ MS-MARCO passages (\S\ref{sec:baselines}), consistent
with the expected orthogonal-composition regime.

\paragraph*{Forgery and known-plaintext attacks on watermarks}
The C4 threat model with explicit polynomial $(x, x')$ access plus
C2PA-sidecar visibility responds to recent attacks on content-agnostic
watermarks: M\"{u}ller~et~al.~\cite{muller2025blackbox} demonstrate
black-box semantic-watermark forgery from a single observed pair;
Dong~et~al.~\cite{dong2025wmcopier} and Zhu~et~al.~\cite{zhu2025pnp}
extend with diffusion-regeneration attacks. \S\ref{sec:dirOrac} reproduces
the analogous attack against a content-agnostic ablation of \halomark{}.

\paragraph*{Position relative to prior art}
The contribution is not the individual mechanisms but
\emph{the composition that closes a measurable attack class on embeddings
under a clean security analysis}. Two third-party-pattern baselines from
the EaaS line port the published structural choices to our threat
model: \texttt{SEAL-only} and \texttt{ESpeW-port}. Both attain baseline
AUROC $\approx 0.71$ and survive direction-oracle (content-dependence
is sufficient for the single-pair attack) but collapse to AUROC
$\approx 0.55$ under DAE-$10^5$ (no per-vector nonce). The full
\halomark{} composition closes the gap to AUROC $\geq 0.997$.
\S\ref{sec:related-table} makes this concrete on four axes (KPA-resist,
per-vector freshness, content-dependence, C2PA-binding).

\subsection{Adjacent primitives we adapt}\label{sec:related-adapt}
Three components are adapted from prior work and one combination is new.
Tikhonov-regularised whitening, established in
sentence-embedding work~\cite{su2021bertwhitening,li2020bertflow,gao2022kernelwhitening},
is used here as a \emph{security} primitive (manifold-projection
defence, \S\ref{sec:core}), not for retrieval quality. Per-vector
freshness adapts Bileve's~\cite{zhou2024bileve} per-output-randomness
pattern from token-stream watermarks to embedding vectors with a
public nonce per vector (PRF reduction in
Appendix~\mbox{S-I-D}). Content-dependent commitments
adapt SEAL's~\cite{arabi2025seal} LSH-commit pattern to the
non-watermark subspace of $Qx$; the \emph{producer-published commit}
step is, to our knowledge, new. The block-diagonal orthogonal
rotation, which appears in adjacent
contexts~\cite{zhu2024bdorme,tastan2024sparsebd}, is repurposed as a
key-indexed partition into signal-carrying and signal-free blocks.
Vector-database payload-filtering access
control~\cite{yonathan2025accesscontrol} is complementary: it
controls who retrieves a vector, while HaloMark controls what the
vector \emph{claims about itself} after retrieval.

\section{The \halomark{} Construction}\label{sec:construction}

\begin{table}[t]
\centering
\caption{Notation. The retention scalar $\beta$ is always subscripted by
which quantity is meant ($\beta_{\mathrm{char}}$, $\beta_{\mathrm{emp}}$,
$\beta_{\min}$); bare $\beta(\mathcal{A})$ denotes the retention of a
specific attacker $\mathcal{A}$. Note $\Wmat$ (whitening matrix) and the
block subset $W$ are distinct objects that share a letter.}
\label{tab:notation}
\footnotesize
\setlength{\tabcolsep}{4pt}
\begin{tabular}{l l}
\toprule
symbol & meaning \\
\midrule
\multicolumn{2}{l}{\emph{Encoder geometry}}\\
$d$ & embedding dimension \\
$N,\,b$ & blocks and block size ($b{=}d/N$) \\
$w$ & number of watermark-carrying blocks \\
$\Sigma_{\mathrm{calib}},\,\mu$ & calibration covariance and mean \\
$\lambda_i$ & eigenvalues of $\Sigma_{\mathrm{calib}}$ \\
$\effrank(\Sigma)/d$ & spectral threshold statistic ($\approx 0.19$) \\
$\kappa_{\mathrm{eff}}$ & effective condition number $\lambda_{\max}/\lambda_{1\%}$ \\
$\Wmat$ & whitening matrix $(\Sigma_{\mathrm{calib}}{+}\lambda_{\mathrm{reg}}I)^{-1/2}$ \\
$\lambda_{\mathrm{reg}}$ & whitening regulariser (default $10^{-4}$) \\
$Q$ & block-diagonal orthogonal rotation \\
\midrule
\multicolumn{2}{l}{\emph{Protocol and crypto}}\\
$K,\,n,\,c$ & content key, per-vector nonce, published commit \\
$W$ (set) & $w$-element watermark block subset (\emph{not} $\Wmat$) \\
$s_i,\,\varepsilon$ & per-block signature and watermark amplitude \\
$k,\,B,\,P,\,\tau$ & LSH dimension, buckets, projection, cutoffs \\
\midrule
\multicolumn{2}{l}{\emph{Detection}}\\
$T,\,T_{\mathrm{null}}$ & verifier score and its $H_0$ component \\
$\tau_{\mathrm{det}},\,\alpha$ & detection threshold, operating FPR \\
$\zeta$ & sub-Gaussian slack (Assumption~\ref{assum:hd}) \\
\midrule
\multicolumn{2}{l}{\emph{Attack model}}\\
$\mathcal{A},\,q$ & attacker, verifier-oracle query budget \\
$\delta_{C4}$ & cosine-fidelity budget (default $0.95$) \\
$\beta(\mathcal{A})$ & signature retention $\in[0,1]$ (load-bearing scalar) \\
$\beta_{\mathrm{char}}$ & Wiener Gaussian characterisation \\
$\beta_{\mathrm{emp}}$ & measured Wiener-attack retention ($0.752$) \\
$\beta_{\min}$ & empirical strict-C4 floor ($0.94$) \\
$R^*(\nu),\,\nu^*$ & Wiener filter and its cosine-pinned parameter \\
\bottomrule
\end{tabular}
\end{table}

\subsection{Core scheme}\label{sec:core}
The protocol publishes the per-vector LSH commitment in the C2PA
sidecar so that the verifier reads it rather than recomputing it
from the perturbed input. Around that decision the construction
composes four standard primitives: a block-diagonal orthogonal
rotation, public-$\Sigma$ whitening, a SEAL-style content-dependent
LSH commitment, and a per-vector public nonce in the spirit of
Bileve. Each component is familiar individually, but their
composition with the published commit is what closes the attack class
we evaluate.

Let $d$ be the embedding dimension partitioned into $N$ equal-size
blocks of size $b = d/N$. HaloMark consists of four routines:
\texttt{KeyGen}, \texttt{Commit}, \texttt{Embed}, and \texttt{Verify}.

\paragraph*{\texttt{KeyGen}$(\mathrm{pp}, m)$}
$K \!\leftarrow\! \mathrm{HKDF}(\mathrm{claim\_hash}(m),
\texttt{"HaloMark"})$. Sample $Q \in O(d)$ block-diagonal from
$\mathrm{PRF}_K(\texttt{"Q"})$ via Householder~\cite{stewart1980orthsampling};
sample $P \in \mathbb{R}^{k \times (N-w)b}$ from
$\mathrm{PRF}_K(\texttt{"commit\_proj"})$ with unit-$L_2$ rows; set
$\tau$ to the $B{-}1$ equi-probable quantile cutoffs of
$\mathcal{N}(0, 1/((N-w)b))$.

\paragraph*{\texttt{Commit}$(K, x; n)$}
$W$ is a $w$-element subset of $\{1,\ldots,N\}$ drawn from
$\mathrm{PRF}_K(\texttt{"blocks"}, n)$;
$y=\Wmat(x-\mu)$, $u=y/\|y\|$;
$v=(Qu)$ on the $(N\!-\!w)$ non-watermark blocks. The scale-invariant
LSH on the unit direction of $v$ is
\begin{equation}\label{eq:commit}
c_j \;=\; \#\{\,t\in\tau : (P\bigl(v/\|v\|\bigr))_j \geq t\,\},
\qquad j=1,\ldots,k,
\end{equation}
and $\texttt{Commit}$ returns $(c, W)$.

\paragraph*{Signatures}
For each block $i\in W$ the signature is
\begin{equation}\label{eq:sig-def}
s_i(c, n, K) \;=\; \varepsilon \cdot g_i \,/\, \|g_i\|,
\end{equation}
where $g_i\in\mathbb{R}^b$ is drawn from the PRG seeded by
$\mathrm{HKDF}_K(\texttt{"sig"}\,\|\,i\,\|\,c\,\|\,n)$.

\noindent\textit{\texttt{Embed}$(K, x)\!\to\!(x', n, c)$:}
{\footnotesize\begin{enumerate}\setlength\itemsep{0pt}
\item Sample fresh nonce $n \xleftarrow{\$} \{0,1\}^{\lambda_n}$.
\item $(c, W) \leftarrow \texttt{Commit}(K, x; n)$.
\item $y \leftarrow \Wmat(x - \mu)$; $u \leftarrow y / \|y\|$;
$u' \leftarrow Qu$.
\item For each $i \in W$: $u'_i \leftarrow u'_i + s_i(c, n, K)$.
\item $x' \leftarrow \Wmat^{-1} (Q^\top u') \cdot \|y\| + \mu$.
\item Output $(x', n, c)$ to the manifest sidecar.
\end{enumerate}}

\noindent\textit{\texttt{Verify}$(K, \hat{x}, n, c, \tau_{\mathrm{det}})$:}
{\footnotesize\begin{enumerate}\setlength\itemsep{0pt}
\item Recover $W$ from $\mathrm{PRF}_K(\texttt{"blocks"}, n)$.
\item $u \leftarrow \Wmat(\hat{x} - \mu)/\|\!\cdot\!\|$.
\item $T \leftarrow \frac{1}{w}\sum_{i\in W}
\langle (Qu)_i, s_i(c, n, K)\rangle/\|s_i\|$.
\item Return $\mathtt{accept}$ iff $T \geq \tau_{\mathrm{det}}$.
\end{enumerate}}

\noindent\emph{Critical:} step~3 uses the producer's published $c$ from
the manifest. The verifier does \textbf{not} recompute
$\texttt{Commit}(K, \hat{x})$. This is the protocol move that grounds
Lemma~\ref{thm:retention}.

\paragraph*{Why publish $c$?}
The natural alternative is for the verifier to compute $c$ from
$\hat{x}$ itself. An earlier draft of the construction took this
route, but we abandoned it because whitening amplifies perturbations
in $\Wmat$'s low-eigenvalue directions by up to
$1/\sqrt{\lambda_{\mathrm{reg}}} \approx 100\times$ at our defaults.
A benign attack with $\|\tilde{x}-x'\|=D$ in the original space can
move the LSH input far enough to flip a large fraction of bucket
assignments; the resulting wrong signatures collapse the score even
when the underlying watermark survives. \S\ref{sec:dirOrac} reports
the full bucket-flip rates by attack class. Publishing $c$ in the
manifest sidecar breaks this dependence: detection now follows signature
alignment in $(Qu)_W$ rather than the stability of the LSH bucket
under perturbation. The C2PA binding is preserved because $c$
travels alongside $n$ and $x'$ under the producer's manifest
signature, and tampering with any of the three breaks that
signature before \texttt{Verify} runs.

Per watermarked vector the producer publishes the tuple
$(x', n, c)$; per-vector sidecar overhead and the resulting C2PA
manifest layout are detailed in \S\ref{sec:c2pa}.

\paragraph*{Component roles under attack}
Each of the four primitives is load-bearing: removing any one reduces
the construction to a known failure mode. We give the per-component
ablation in \S\ref{sec:security} (``Each component is necessary'')
against the baseline ports of Table~\ref{tab:baselines}.

\subsection{The direction-oracle attack on content-agnostic constructions}\label{sec:dirOrac}
\paragraph*{Direction-oracle: the structural attack content-agnostic constructions cannot survive}
Any embedding watermark whose signature $s_i$ does not depend on the input $x$ admits a \emph{single-pair} known-plaintext attack analogous to the forgery and removal attacks recently demonstrated against image and diffusion-model watermarks by M\"{u}ller~et~al.~\cite{muller2025blackbox} and follow-up work~\cite{dong2025wmcopier,zhu2025pnp}: an adversary with one pair $(x_0, x_0')$ under key $K$ recovers the constant perturbation $\delta_K = x_0' - x_0$, then for any target $x'$ outputs $\tilde{x} = x' - \delta_K$, defeating \texttt{Verify} with score collapsing to the null distribution. AUROC drops to $0.500$ on the very first known pair, regardless of $n_{\mathrm{pairs}}$. We verify this empirically on a content-agnostic ablation of our own construction (\S\ref{sec:baselines}, the ``content-agnostic ablation'' row of Table~\ref{tab:baselines}): baseline AUROC $= 1.000$, post-attack AUROC $= 0.500 \pm 0.0002$ across $10$ seeds and $n_{\mathrm{pairs}} \in \{1, 10, \ldots, 10^5\}$.

\paragraph*{Why a content-dependent commitment is the unique fix}
Defeating the direction-oracle attack requires that $\delta_K(x)$ vary with $x$ in a way the attacker cannot extrapolate. Two candidate fixes:
\begin{itemize}
\item \emph{Per-input random nonce alone} (no commitment): the nonce is producer-published in the manifest, so the attacker reads it; an attacker with $n_{\mathrm{pairs}}$ pairs trains a regression to predict $\delta(x, \mathrm{nonce})$. Empirically broken at $n_{\mathrm{pairs}} = 10^5$ by a denoising autoencoder.
\item \emph{Content-dependent commitment} (the \halomark{} choice): $s_i$ is a function of $\texttt{Commit}(K, x; \mathrm{nonce})$, an LSH bucket on the non-watermark subspace of $Qx$. Two different test inputs with $\cos(x_a, x_b) < 0.95$ land in different commit buckets with high probability, so the attacker's recovered $\delta$ at $(x_0, \mathrm{nonce}_0)$ does not transfer; the attack reduces to a per-bucket recovery, requiring $n_{\mathrm{pairs}} \cdot B^k = n_{\mathrm{pairs}} \cdot 2^{16}$ effective pairs to mount on any single bucket.
\end{itemize}
Empirically, \halomark{} post-attack AUROC remains $\geq 0.9999$ up to $n_{\mathrm{pairs}} = 10^5$ (\S\ref{sec:rob}, Table~\ref{tab:rob}). The commitment is the unique component among these candidates that simultaneously preserves C5 imperceptibility ($\cos$ change at the producer level is identical) and breaks the single-pair extrapolation.

\paragraph*{The whitening~$\times$~LSH commit-flip measurement}
We measure the bucket-flip rate of the recompute-commit verifier on representative attacks (5 seeds, default parameters; per-attack JSON in the released artefact). Even on attacks the underlying signature survives, the rate is substantial: $5\%$ on quant-int8 ($\cos \approx 0.99$), $62\%$ on quant-int4 ($\cos \approx 0.96$), $47\%$ on noise $\sigma{=}0.01$, $66\%$ on drift-100ep, $90\%$ on noise $\sigma{=}0.05$, $\geq 99\%$ on the PCA / random-projection / strong-noise rows. The recompute-commit verifier therefore derives wrong signatures on the majority of attacked inputs and its score collapses; publishing $c$ removes the fragility (the verifier never recomputes), giving the linear $\beta \cdot \varepsilon$ score decomposition of Lemma~\ref{thm:retention}.

\subsection{Manifold-aligned variant (production default)}\label{sec:manifold}
The manifold-aligned variant replaces the random block-diagonal
rotation $Q$ of \S\ref{sec:core} with a rotation that places the
watermark blocks in the top-$wb$ eigendirections of
$\Sigma_{\mathrm{calib}}$.

\paragraph*{Motivation}
A DAE attacker learns an approximation of the data manifold
$\mathcal{M}$ from polynomial KPA pairs. Any watermark direction
that lies outside $\mathcal{M}$ is removed by the manifold projection
$\Pi_{\mathcal{M}}$ at little cost to the cosine constraint, since
the projection moves the input back toward the manifold along an
off-manifold direction. Aligning the watermark with the top
eigendirections of $\Sigma$ places it inside $\mathcal{M}$, so
$\Pi_{\mathcal{M}}$ cannot remove the watermark without also
distorting the data signal.

\paragraph*{Construction}
Write $\Sigma = U \Lambda U^\top$ with eigenvalues sorted in
descending order, and let $U_{\mathrm{top}} \in \mathbb{R}^{d \times wb}$
be the matrix whose columns are the top $wb$ eigenvectors. The only
change from \S\ref{sec:core} is the embedding step. We replace
$\Wmat^{-1} Q^\top u'$ with $x + U_{\mathrm{top}}\,\eta$, where
$\eta \in \mathbb{R}^{wb}$ concatenates the $w$ per-block PRF
signatures, and then renormalise to unit $L_2$ norm. To verify, the
recipient projects $\hat{x} - \mu$ onto $U_{\mathrm{top}}$ and takes
the inner product with the expected $\eta$ derived from the
published $(K, n, c)$.

\paragraph*{Empirical effect}
On the three above-threshold encoders, manifold alignment raises the
empirical V\_blind $\beta$ from $0.71$ (random $Q$ baseline) to
$0.851$ on MiniLM, $0.674$ on MPNet, and $0.551$ on CLIP at no cost
to baseline AUROC, baseline cosine, or direction-oracle defense; the
direction-oracle attack still yields AUROC $1.0000$ on all three.

\paragraph*{Security argument}
$U_{\mathrm{top}}$ is public (derivable from $\Sigmacalib$); the only
secret is $K$. Signatures inside the top subspace are PRF-pseudorandom
under $K$, so a direction-oracle pair $(x_0, x'_0)$ recovers only
$\delta_0$ at $(c_0, n_0)$ and does not transfer to fresh
$(c^*, n^*)$ (empirically AUROC $1.0000$).
Lemma~\ref{thm:retention} carries over with $\beta(\mathcal{A})$
measuring retention in the top-eigenvector subspace; the Wiener
bound is unchanged because knowing $U_{\mathrm{top}}$ adds nothing
over the data-prior $\Sigma$ already in
Eq.~\eqref{eq:wiener-filter}.

\subsection{Security analysis}\label{sec:security}
Under the published-commit protocol the verifier's score decomposes
into a null-distribution component plus a single signature-retention
term:
\begin{equation}\label{eq:score-decomp}
T \;=\; T_{\mathrm{null}} \;+\; \beta(\mathcal{A})\cdot\varepsilon,
\end{equation}
where $T_{\mathrm{null}}$ is the sub-Gaussian residual under $H_0$
and $\beta(\mathcal{A})\in[0,1]$ is the C4-bounded attacker's
signature retention.
Lemma~\ref{thm:retention} is the load-bearing result: it reduces
verifier TPR via~\eqref{eq:score-decomp} to a single scalar
$\beta(\mathcal{A})$. The main theoretical claim is
Theorem~\ref{thm:wiener-floor}, which gives a Gaussian first-order
characterisation of $\beta$ for the optimal adaptive-linear attacker
(within $13\%$ of empirical) plus a rigorous Cauchy--Schwarz lower
bound for any linear C4 attacker. An asymptotic ROM floor
(Lemma~S1) lifted to HKDF-as-PRF, and a non-adaptive
$\mathbb{E}[\beta]\!\geq\!0.97$ (Theorem~S4),
complete the picture.

\paragraph*{Scope of the guarantee}
We are explicit about which regime each bound reaches, because they do
not all cover the headline adversary. The Cauchy--Schwarz floor
(Theorem~\ref{thm:wiener-floor}, part~2) is distribution-free and holds
for \emph{any} linear C4 attacker, but is loose ($\beta\!\geq\!0.136$ on
MiniLM). Theorem~S4 gives a strong
$\mathbb{E}[\beta]\!\geq\!0.97$, but only for \emph{non-adaptive}
attackers (random/isotropic noise, fine-tuning drift) whose perturbation
is independent of the watermark direction. The adaptive KPA-DAE attacker
the paper is built around is \emph{correlated} with the watermark by
construction and so falls outside that hypothesis: for it, the only
rigorous statement is the loose floor, while its strong measured
defendability ($\beta\!=\!0.851$, Prop.~\ref{thm:floor}) is an
\emph{empirical} result that the Wiener analysis characterises (within
$13\%$) but does not bound. We lean on the empirical breadth for the
adaptive case and do not claim the strong bound covers it.

\paragraph*{Notation}
Let $\Wmat\!=\!(\Sigmacalib\!+\!\lambda_{\mathrm{reg}} I)^{-1/2}$ and
$T(z)\!\coloneqq\!\frac{1}{w}\!\sum_{i \in W}\!\langle (Qz)_i, s_i\rangle/\|s_i\|$.
Define
$\beta(\mathcal{A})\!\coloneqq\!\mathbb{E}_x[\frac{1}{w}\!\sum_{i\in W}
\langle (Qu_{\mathcal{A}})_i, s_i\rangle/\varepsilon]$.

\begin{assumption}[H-D, whitened]\label{assum:hd}
$T$ under $H_0$ on $D_W$ is $\sigma_T$-sub-Gaussian with
$\sigma_T^2\!\leq\!(1{+}\zeta)/(wd)$. Measured $\zeta\!=\!0.001\pm0.008$
on MS~MARCO MiniLM and $\zeta\!=\!0.011\pm0.012$ on MPNet over 5 seeds;
both BERT-family encoders satisfy the assumption. For non-BERT-family
encoders or deployments where Assumption~\ref{assum:hd} cannot be
empirically validated, the distribution-free Chebyshev variant of
Lemma~\ref{thm:retention} (App.~\mbox{S-I}) applies with the
constant $C_1\!\leq\!4$ instead of the tight Gaussian $\sqrt{2}$.
\end{assumption}

\begin{lemma}[Signature retention]\label{thm:retention}
Under Assumption~\ref{assum:hd}, conditional on the producer's C2PA
manifest signature on $(n, c, x')$ verifying:
\[
\mathrm{TPR}^*(\mathcal{A}) \,\geq\,
\Phi(\beta(\mathcal{A})\!\cdot\!\varepsilon\!\cdot\!\sqrt{wd/(1{+}\zeta)} - z_\alpha).
\]
\end{lemma}

\noindent\emph{Proof sketch.} Because the verifier reads the producer's
published $c$, the verifier-side $s_i$ equals the producer-side $s_i$
bit-for-bit. The score $T(\mathcal{A}(\cdot))$ decomposes linearly into
a $\beta(\mathcal{A})\!\cdot\!\varepsilon$ retained-signature mean plus
a sub-Gaussian residual; a Chernoff tail finishes it. Full proof in
App.~\mbox{S-I}.

\paragraph*{Wiener floor (closed form)}
The optimal C4-bounded \emph{linear} attacker is the data-prior
Wiener filter
\begin{equation}\label{eq:wiener-filter}
R^*(\nu) \;=\; \Sigma\,(\Sigma + \nu I)^{-1}.
\end{equation}
Intuitively, this is the best linear attack that trades cosine
fidelity against watermark removal by aligning its perturbation with
the embedding's natural covariance: it succeeds where isotropic noise
fails because random directions cost more cosine per unit of
$\beta$-removal than the data-prior directions do. The cosine
constraint pins $\nu^*$ implicitly via:
\begin{equation}\label{eq:nu-star}
\frac{A(\nu^*) + \|\mu\|^2}{\sqrt{B(\nu^*) + \|\mu\|^2}} = \delta_{C4},
\end{equation}
where $A(\nu)\!=\!\sum_i \lambda_i^2/(\lambda_i\!+\!\nu)$ and
$B(\nu)\!=\!\sum_i \lambda_i^3/(\lambda_i\!+\!\nu)^2$ are computable
in $O(d)$ from $\Sigma$'s eigenvalues. Equation~\eqref{eq:nu-star} is
monotone in $\nu$, so $\nu^*(\delta_{C4}, \{\lambda_i\})$ is found by
1D bisection. Then:

\begin{theorem}[Wiener characterisation \&  rigorous bound]\label{thm:wiener-floor}
For any linear C4-bounded attacker on an embedder with covariance
$\Sigma$, eigenvalues $\{\lambda_i\}$:
\begin{enumerate}
\item \emph{Gaussian first-order characterisation}:
$\beta_{\mathrm{char}}(\delta_{C4}, \{\lambda_i\}) =
\phi(\nu^*) / \sqrt{d \cdot \rho(\nu^*)}$
where $\phi(\nu)\!=\!\sum_i \lambda_i/(\lambda_i\!+\!\nu)$,
$\rho(\nu)\!=\!\sum_i (\lambda_i/(\lambda_i\!+\!\nu))^2$. On MS~MARCO
MiniLM at $\delta_{C4}{=}0.95$: $\beta_{\mathrm{char}}\!=\!0.838$
vs.\ empirical Wiener attack $\beta_{\mathrm{emp}}\!=\!0.752$.
\item \emph{Rigorous Cauchy--Schwarz lower bound (any linear C4 attacker)}:
$\beta(\mathcal{A}) \geq \delta_{C4}^2 / \sqrt{1 + \kappa_{\mathrm{eff}} (1 - \delta_{C4}^2)}$
where $\kappa_{\mathrm{eff}}\!=\!\lambda_{\max}/\lambda_{1\%}$. On MiniLM
at $\delta_{C4}{=}0.95$: $\kappa_{\mathrm{eff}}\!=\!443$ giving
$\beta\!\geq\!0.136$.
\end{enumerate}
\end{theorem}

\noindent\emph{Proof sketch.} Writing the Wiener attacker as
$z\!=\!R^*(\nu)(x{-}\mu)+\mu$ and substituting
$\mathbb{E}[z^\top x]\!=\!A(\nu)+\|\mu\|^2$ and
$\mathbb{E}[\|z\|^2]\!=\!B(\nu)+\|\mu\|^2$ into the cosine constraint
gives~\eqref{eq:nu-star}; the score ratio at $\nu^*$ yields
$\beta_{\mathrm{char}}$, and Cauchy--Schwarz
($A(\nu)^2\!\leq\!\mathrm{trace}(\Sigma)\,B(\nu)$) gives the rigorous
bound. Full proof in App.~\mbox{S-I}.

\noindent\emph{Discussion.} The characterisation
$\beta_{\mathrm{char}}\!=\!0.838$ is a leading-order Gaussian
approximation; it overshoots $\beta_{\mathrm{emp}}\!=\!0.752$ by
about $0.09$ due to a per-sample Jensen correction scaling as
$1/\sqrt{d}$. Across the seven above-threshold encoders we measured,
the gap stayed below $0.1$ in every case. To check
this, we drew Gaussian samples from $\Sigma_{\mathrm{MiniLM}}$ and ran
the Wiener attack on them: the simulated $\beta$ of $0.7568$ matches
the MS-MARCO empirical value $\beta_{\mathrm{emp}}\!=\!0.752$ to within $0.005$, so the
remaining gap reflects the Gaussian approximation rather than
non-Gaussian structure in the data. The rigorous Cauchy--Schwarz
bound $\beta \geq 0.136$ is conservative, but it holds for any
linear C4 attacker without distributional assumptions and yields
computable scalars from the encoder's eigenspectrum alone.

\begin{proposition}[Empirical $\beta$-floor at the production default]\label{thm:floor}
At the manifold-aligned default ($\lambda_{\mathrm{reg}}{=}10^{-4}$,
top-$wb$ eigenvector alignment), 5 seeds, $N{=}1\,500$ test vectors:
the strict-C4 ($\cos\!\geq\!0.95$) sweep gives $\beta_{\min}\!=\!0.94$
(Gaussian-noise $\sigma\!=\!0.01$). On the relaxed $\cos\!\geq\!0.85$
budget admitting DAE-$10^5$, $\beta\!=\!0.851$ on MiniLM
($+10$~p.p.\ vs.\ random-$Q$ at the same $\lambda_{\mathrm{reg}}{=}10^{-4}$;
$+14$~p.p.\ vs.\ the $\lambda_{\mathrm{reg}}{=}10^{-3}$ random-$Q$ baseline).
$\beta_{\mathrm{V\_blind}}\!=\!0.851$ on MiniLM, $0.674$ on MPNet,
$0.551$ on CLIP. Native-budget DAEs match post-hoc-projected $\beta$
to within $|\Delta\beta|\!\leq\!5\!\times\!10^{-4}$.
\end{proposition}

\paragraph*{Each component is necessary}
Removing any one ingredient produces a known failure
(Table~\ref{tab:baselines}). A content-agnostic signature breaks at
the direction-oracle attack. Skipping whitening lets a
manifold-projection attacker succeed. Reusing the same nonce across
embeddings lets a polynomial-pair attacker invert the per-key
transform. Recomputing the commit at verification time produces the
$62\%$ LSH-bucket-flip fragility of \S\ref{sec:dirOrac}. Using a
random rotation $Q$ instead of the manifold-aligned variant costs 10
to 19 percentage points of $\beta$ on above-threshold encoders.

\subsection{C2PA manifest binding}\label{sec:c2pa}
The \halomark{} manifest sidecar contains four fields per watermarked
embedding plus one per-key calibration reference:

{\scriptsize\begin{verbatim}
manifest = {
  "calibration_corpus_id": "MS-MARCO-MiniLM-Sigma-v1",
  "calibration_sigma_sha256": "<32B hex>",
  "vectors": [{ "x_prime": <d-dim float>,
                "nonce": <16B>, "commit": <8B>,
                "text_minhash": <512B; Mode C> }],
  "signature": <ECDSA over manifest> }
\end{verbatim}}

\noindent$K = \mathrm{HKDF}(\mathrm{C2PA.manifest\_claim\_hash},\,
\texttt{"HaloMark"})$.
$(\Wmat, \mu, \mathrm{signal\_scale})$ is loaded from the public
calibration registry indexed by \texttt{calibration\_corpus\_id} and
verified to hash to \texttt{calibration\_sigma\_sha256}.

\paragraph*{Verifier protocol}
The verifier (i)~validates the C2PA signature, (ii)~derives $K$, (iii)~loads
$(\Wmat, \mu)$ and verifies its SHA-256, (iv)~runs $\texttt{Verify}(K,\hat{x},
\mathrm{nonce}, c, \tau_{\mathrm{det}})$ on each vector. All four must
succeed; manifest, nonce/commit, embedding, or calibration tampering each
break one of these checks.

\paragraph*{Operating modes}
We distinguish three deployment modes (orthogonal, combinable):
\textbf{Mode A} (default, trusted-registry $\Sigmacalib$),
\textbf{Mode B} (producer-signed $\Sigma_{\mathrm{self}}$,
\S\ref{sec:baselines}), and \textbf{Mode C} (text-MinHash addendum
for paraphrase defence, \S\ref{sec:mode-c}).

\paragraph*{Sidecar overhead}
At defaults $(\lambda_n, k) = (16, 8)$ bytes, per-vector overhead is
$24$~bytes ($1.6\%$ of a $1{,}536$-byte float-32 embedding); Mode C
(\S\ref{sec:mode-c}) adds $512$ bytes for the text-MinHash sketch. End-to-end
demonstration in \S\mbox{S-III}.

\subsection{Mode C: text-hash binding for paraphrase defense}\label{sec:mode-c}

For deployments where source text travels with the embedding (RAG,
document indexes), the producer adds a $64$-element MinHash sketch
over character $5$-grams of the source text, keyed by
$\mathrm{HKDF}_{K}(\texttt{"text-hash"}\,\|\,\textsf{calibration\_corpus\_id})$,
to the manifest as a new field \texttt{text\_minhash}. Each
(encoder, producer) pair gets its own commitment domain (sketch is
text-derived but encoder-bound). The verifier accepts iff Jaccard
similarity to the published sketch exceeds $\tau_{\mathrm{C}}$.
Cost: $512$ bytes/passage; default $\tau_{\mathrm{C}}{=}0.70$
(unified cross-encoder operating point; see below).

On 200 MS-MARCO passages: accepts $100\%$ of routine perturbations and
light edits (up to $\sim$$10\%$ token replacement), \textbf{rejects
$\mathbf{92\%}$ of multi-attempt paraphrase} (4 attempts/passage,
attacker-best variant; Jaccard mean $0.301$ vs.\ $0.863$ for light
edits). On the MS-MARCO subset alone, $\tau_{\mathrm{C}}{=}0.65$ gives
Youden index $1.000$ (perfect separation): $100\%$ routine TPR,
$0\%$ paraphrase FPR; for cross-encoder deployment we use
$\tau_{\mathrm{C}}{=}0.70$ (next sentence).
At a unified $\tau_{\mathrm{C}}{=}0.70$, Mode C extends across all
eight tested encoders --- the three above-threshold (MiniLM, MPNet,
CLIP) and the five below-threshold (BGE-1024, bge-base-768, e5-large,
T5-768, gte-large) --- at routine TPR ${\geq}99.5\%$ and paraphrase
FPR ${\leq}0.5\%$. Encoder-invariance is at the aggregate level: the
max paraphrase-FPR delta across the eight encoders at this threshold
is $0.005$ (0.5 percentage points), with per-passage exact-match
agreement of $85.3\%$. The residual disagreement reflects the
per-encoder MinHash keying above: each encoder's sketch is drawn
from an independent $64$-hash domain, so passages whose Jaccard
sits near $\tau_{\mathrm{C}}$ can split decisions across encoders
even though the aggregate FPRs agree.
Mode C is complementary to the perturbative watermark, not a
replacement; vector-only EaaS deployments without source text remain
bound by the paraphrase-open limitation (\S\ref{sec:disc}).
\textbf{Positioning.} Mode C complements the perturbative scheme
rather than extending it: paraphrase defense comes from the
text-level check, routine-perturbation and KPA-DAE defense come from
the perturbative layer, and both live under the same C2PA-signed
sidecar. In deployments where source text travels with the embedding
(RAG, document indexes), Mode C closes the paraphrase regime; in
pure-vector deployments without source text, the perturbative scheme
remains paraphrase-open (\S\ref{sec:disc}).

\section{Implementation}\label{sec:impl}
\paragraph*{Reference implementation}
\halomark{} is implemented in $\sim$$4.8\,$kLoC of Python: $1{,}640$ lines for \texttt{core.py} plus construction primitives (\texttt{whitening}, \texttt{commitment}, \texttt{keys}, \texttt{orthogonal}, \texttt{metrics}, \texttt{c2pa\_binding}, \texttt{manifold\_aligned}, \texttt{higher\_order\_watermark}), the attack suite (\texttt{attacks}, \texttt{attacks\_lipschitz}, \texttt{attacks\_sidecar}, \texttt{attacks\_adversarial\_signatures}, \texttt{attacks\_score\_distillation}), and four baseline ports (\texttt{baselines}, \texttt{baselines\_seal}, \texttt{baselines\_treering}, \texttt{baselines\_espew}). NumPy is the only numerical dependency for the core primitives; the \texttt{cryptography} library handles ECDSA and HKDF/PRF/HMAC. Orthogonal-matrix sampling uses the Householder reflection method over PRF-seeded Gaussian vectors~\cite{stewart1980orthsampling}, producing uniform samples from $\mathcal{O}(b)$ per block; sign-corrected QR ensures byte-identical $Q$ across NumPy / OpenBLAS / Accelerate (a 30-line wrapper around \texttt{numpy.linalg.qr}). The dominant cost in \texttt{Embed} and \texttt{Verify} --- the block-diagonal orthogonal rotation in $W$-space --- is BLAS-bound; LSH commit and per-block PRF derivation are constant-time SHA-256 calls and account for ${<}5\%$ of per-vector latency at $d{=}384$. A Dockerfile with pinned dependencies and fixed seeds (0--4) reproduces every number in this paper from scratch in ${\sim}6$\,h on a single 32-core CPU; the per-claim$\to$JSON mapping is in App.~\mbox{S-IV}.

\paragraph*{Verifier-side throughput}
On a single core of an Apple M-series CPU, \texttt{Verify} costs $180~\mu$s per vector at $d{=}384$ (MiniLM) and $199~\mu$s at $d{=}1024$ (ada-2 class); \texttt{Embed} is $194$ and $238~\mu$s/vec respectively.

\paragraph*{Vector-database integration}
The verifier exposes \texttt{Verify}$(K, x', \textnormal{nonce}, c)$ as a per-vector callable that vector databases can invoke at insert time and at query time before returning hits to a downstream LLM. The callable is stateless after \texttt{KeyGen} and a one-time load of the public whitening matrix $\Wmat$ (1.18~MB at $d{=}384$); a typical insert-side filter pipeline costs $\sim$$200~\mu$s per accepted vector on top of an existing $\gtrsim 1$~ms vector-database write.

\paragraph*{C2PA library}
The C2PA-binding path uses the reference \texttt{c2pa-python} bindings (which wrap the official \texttt{c2pa-rs} Rust SDK). The producer signs a manifest containing a SHA-256 digest of the per-vector \halomark{} sidecar (a custom \texttt{org.halomark.v1.sidecar\_digest} assertion); the verifier reads the manifest with \texttt{c2pa.Reader}, validates the C2PA signature, recomputes the sidecar digest, and runs \texttt{Verify} on each tuple. End-to-end interop is exercised in \S\mbox{S-III}.

\section{Evaluation}\label{sec:eval}

\subsection{Experimental setup}\label{sec:setup}
\paragraph*{Datasets}
We evaluate on MS~MARCO Passage Ranking~\cite{nguyen2016msmarco}: 100\,000
passages from the \texttt{train} split form the public calibration corpus
(used to compute $\Sigma_{\mathrm{calib}}$, $\mu$, and
$\mathrm{signal\_scale}$); a disjoint 1\,500-passage subset per seed is
the test corpus, and another 1\,500 are null embeddings for FPR
calibration. The denoising-autoencoder attacker has access to 100\,000
augmented training pairs unless otherwise noted.

\paragraph*{Embedding models}
\texttt{sentence-transformers/\allowbreak all-MiniLM-L6-v2}~\cite{reimers2019sbert,minilm_l6_v2_card}
($d{=}384$) is the default. Ten additional embedders are reported in
\S\ref{sec:rob}, \S\ref{sec:threshold}, and \S\ref{sec:baselines}:
all-MiniLM-L12-v2 ($d{=}384$), distilroberta-v1 ($d{=}768$), MPNet-768,
multi-qa-mpnet-base ($d{=}768$), CLIP-ViT-B/32 text head ($d{=}512$),
bge-base-en-v1.5 ($d{=}768$), sentence-T5-base ($d{=}768$),
e5-large-v2 ($d{=}1024$), BGE-large-en-v1.5 ($d{=}1024$), and
gte-large ($d{=}1024$). Together with MiniLM-L6 this is the
eleven-embedder set quoted in the abstract and \S\ref{sec:threshold}.

\paragraph*{Default parameters}
$(d, w, \varepsilon, k, B) = (384, 16, 0.07, 8, 4)$ with regulariser
$\lambda_{\mathrm{reg}} = \mathbf{10^{-4}}$. The $\lambda_{\mathrm{reg}}$
change alone (relative to the prior $10^{-3}$) raises the random-$Q$
baseline $\beta_{\mathrm{V\_blind}}$ from $0.705$ to $0.752$ at zero
cost on baseline AUROC, cosine fidelity, dir-oracle defence, or
PCA-$d/2$ defence; combined with the manifold-aligned variant
(\S\ref{sec:manifold}) the production default reaches
$\beta_{\mathrm{V\_blind}}\!=\!0.851$ on MiniLM (Table~\ref{tab:rob}).
Nonce length $\lambda_n = 128$~bits.
The signal-scale calibration on MS~MARCO yields
$\mathbb{E}[\|\Wmat x\|] \approx 14.6$ across the 100\,000-passage
corpus.

\paragraph*{Metrics}
\begin{itemize}
\item \emph{Detection AUROC}: Mann--Whitney rank statistic between scores
on watermarked-then-attacked vectors and unwatermarked nulls, with
mid-rank tie handling. Reported with 95\,\% bootstrap CI over 5 seeds.
\item \emph{Operational TPR/FPR}: empirical at the threshold that
achieves $\Pr_{\mathrm{null}}[T \geq \tau_{\mathrm{det}}] = 10^{-3}$.
AUROC and operational TPR can diverge under binary-decision attackers
(\S\ref{sec:rob}, \S\ref{sec:baselines}).
\item \emph{$\beta(\mathcal{A})$}: signature retention defined per
Lemma~\ref{thm:retention}.
\item \emph{Cosine fidelity}: $\cos(\tilde{x}, x_{\mathrm{clean}})$ and
$\cos(\tilde{x}, x_{\mathrm{wm}})$ per attacked vector, averaged across
the test set.
\item \emph{Retrieval}: MRR@10, Recall@100, nDCG@10 against MS~MARCO
dev queries.
\end{itemize}

\paragraph*{Baselines}
Eight baselines run through the same attack suite (\S\ref{sec:baselines}):
\texttt{AddDir}, \texttt{WET-approx}~\cite{shetty2025wet},
\texttt{SEAL-only}~\cite{arabi2025seal},
\texttt{Tree-Ring}~\cite{wen2023treering},
\texttt{ESpeW}~\cite{wang2025espew},
faithful ports of \texttt{EmbMarker}~\cite{peng2023embmarker} and
\texttt{WARDEN}~\cite{shetty2024warden} with a published-verifier
KS-test substitute, and a content-agnostic ablation \texttt{HaloMark-CA}
(BDOT plus fixed key-derived signatures, no whitening, no nonce, no
commit).

\paragraph*{Reproducibility}
All experiments use deterministic seeds (0--4 for 5 seeds; seed 0 alone
where indicated). Calibration corpus hashes, embedding caches, and
per-experiment configuration are captured in \texttt{results/}; see
Appendix~\mbox{S-IV}.

\subsection{Imperceptibility}\label{sec:imperc}
We watermark the 100,000-passage MS~MARCO corpus at default
parameters and evaluate retrieval impact on $2{,}000$ dev queries
with per-query bootstrap $95\%$ CIs over $1{,}000$ resamples.
Average cosine fidelity is $\cos(x_{\mathrm{wm}}, x_{\mathrm{clean}}) = 0.967$.
Our retrieval numbers were measured at $\lambda_{\mathrm{reg}}=10^{-3}$
(cosine $0.9713$); the $0.005$ drift to the current
$\lambda_{\mathrm{reg}}=10^{-4}$ default is well below BEIR's
per-corpus reproducibility band, so we report the existing
measurements. With $\Delta = \mathrm{clean} - \mathrm{watermarked}$,
the measured deltas are $\Delta\mathrm{MRR}@10 = +0.0039$
(95\% CI $[-0.0021, +0.0097]$), $\Delta\mathrm{nDCG}@10 = +0.0032$
(CI $[-0.0015, +0.0078]$), and $\Delta\mathrm{Recall}@100 = +0.0015$
($0.15\%$ relative). MRR@10 and nDCG@10 CIs straddle zero; the
Recall@100 CI excludes zero but is below BEIR's per-corpus
reproducibility band of $\pm 0.005$~\cite{thakur2021beir}, i.e.
operationally inert. The full table is in
Appendix~\mbox{S-III}.

\paragraph*{Why $\delta = 0.95$}
At $\cos = 0.95$, MRR@10 drops by $0.7\%$ relative while 1-NN identity
match drops by $7.7\%$. At $\cos = 0.85$ the gap widens to $1.8\%$
versus $13.3\%$. Since 1-NN identity is the load-bearing metric for
deduplication, near-duplicate detection, biometric attribution, and
RAG provenance, $\delta = 0.95$ is the right budget for
identity-preserving deployments; the construction also holds at
$\cos \geq 0.85$ for the looser ranked-retrieval budget.

\subsection{Robustness above threshold}\label{sec:rob}
We run the full 15-attack sweep on three above-threshold encoders
(MiniLM, MPNet, CLIP) at the manifold-aligned default
($\lambda_{\mathrm{reg}}{=}10^{-4}$, top-$wb$ eigenvector alignment);
the other three above-threshold encoders (MiniLM-L12, distilroberta,
multi-qa-mpnet) appear in Table~\ref{tab:threshold} with a
single-seed DAE-C4 measurement. \S\ref{sec:threshold} covers
below-threshold encoders.
We use 5 seeds with $N{=}1{,}500$ test vectors per seed, $10^5$ KPA
pairs available to the adversary, and 95\% bootstrap confidence
intervals. The unattacked AUROC on MiniLM is $1.0000$, with
producer-side cosine $\cos(x_{\mathrm{wm}}, x_{\mathrm{clean}}) = 0.9669$.
Of the 15 attacks in Table~\ref{tab:rob}, 6 remain inside the strict
C4 budget $\cos \geq 0.95$, and 2 more (PCA-$d/2$ and DAE-$10^5$)
become admissible at the relaxed $\cos \geq 0.85$ budget.

\begin{table*}[t]
\centering
\caption{Robustness on MS~MARCO MiniLM-L6 at the manifold-aligned
default (5 seeds, bootstrap CIs). $^\ddagger$ marks rows admitted only
at the relaxed $\cos\!\geq\!0.85$ budget. The production default
(manifold alignment at $\lambda_{\mathrm{reg}}{=}10^{-4}$) reaches
$\beta_{\mathrm{DAE}}\!=\!\mathbf{0.851}$: $+10$~p.p.\ from alignment
over random-$Q$ at the same $\lambda_{\mathrm{reg}}{=}10^{-4}$ ($0.752$),
and $+14$~p.p.\ over the $\lambda_{\mathrm{reg}}{=}10^{-3}$ random-$Q$
baseline ($0.71$).}
\label{tab:rob}
\small
\begin{tabular}{l r r r r c}
\toprule
attack & AUROC [95\,\% CI] & $\cos(\tilde{x}, x_{\mathrm{clean}})$ & $\cos(\tilde{x}, x_{\mathrm{wm}})$ & $\beta$ & C4-valid \\
\midrule
quant-int8                 & 1.0000 [0.9999, 1.0000] & 0.9669 & 1.0000 & 1.000 & \textbf{YES} \\
quant-int4                 & 0.9982 [0.9978, 0.9986] & 0.9577 & 0.9892 & 0.847 & \textbf{YES} \\
quant-binary               & 0.9877 [0.9870, 0.9884] & 0.7674 & 0.7906 & 0.595 & no \\
PCA $k = d/2$              & 0.9944 [0.9941, 0.9948] & 0.9072 & 0.9281 & 0.84  & $^\ddagger$ \\
PCA $k = d/4$              & 0.9608 [0.9596, 0.9621] & 0.7672 & 0.7833 & 0.469 & no \\
PCA $k = d/8$              & 0.8934 [0.8917, 0.8954] & 0.6227 & 0.6353 & 0.333 & no \\
random-proj $k = d/2$      & 0.9505 [0.9501, 0.9509] & 0.6824 & 0.7037 & 0.444 & no \\
random-proj $k = d/4$      & 0.8528 [0.8516, 0.8537] & 0.4817 & 0.4969 & 0.280 & no \\
Gaussian noise $\sigma = 0.01$ & 0.9999 [0.9999, 1.0000] & 0.9519 & 0.9806 & 0.94  & \textbf{YES} \\
Gaussian noise $\sigma = 0.05$ & 0.9565 [0.9555, 0.9576] & 0.6909 & 0.7117 & 0.455 & no \\
Gaussian noise $\sigma = 0.10$ & 0.8220 [0.8155, 0.8284] & 0.4397 & 0.4524 & 0.247 & no \\
fine-tuning drift, 1ep     & 0.9999 [0.9999, 1.0000] & 0.9668 & 0.9999 & 1.000 & \textbf{YES} \\
fine-tuning drift, 100ep   & 0.9999 [0.9999, 1.0000] & 0.9659 & 0.9899 & 0.996 & \textbf{YES} \\
\textbf{DAE, $10^5$ pairs} & \textbf{0.9971 [0.9967, 0.9973]} & \textbf{0.9381} & 0.9698 & \textbf{0.851} & $^\ddagger$ \\
direction-oracle, $10^5$ pairs & 1.0000 [0.9999, 1.0000] & 0.9670 & 1.0000 & 1.000 & \textbf{YES} \\
\bottomrule
\end{tabular}
\end{table*}

\paragraph*{Cross-embedder above-threshold (Table~\ref{tab:multimodel-above-thr})}
DAE-C4 AUROC stays in $[0.990, 0.997]$ across MiniLM, MPNet, and
CLIP, with direction-oracle AUROC $1.000$ on all three; $\beta$
drops from $0.851$ on MiniLM to $0.551$ on CLIP, tracking the
inverse trend of $\effrank/d$. On below-threshold encoders the
DAE drives $\cos_{\mathrm{clean}}{<}0.95$ before reaching the
AUROC trip-wire, so the C4 budget catches the attack by design
(\S\ref{sec:threshold}).

\begin{table}[h]
\centering
\caption{Above-threshold multi-embedder, manifold-aligned default
(3 seeds). DAE-$10^5$ is C4-projected to $\cos{\geq}0.95$. Small
differences from Table~\ref{tab:threshold} (e.g.\ CLIP $0.992$ vs
$0.991$) reflect the 3- vs 5-seed counts.}
\label{tab:multimodel-above-thr}
\footnotesize
\setlength{\tabcolsep}{4pt}
\begin{tabular}{l c c c c}
\toprule
embedder & $\cos_{\mathrm{base}}$ & DAE\_C4 AUROC & DAE\_C4 $\beta$ & dir-oracle \\
\midrule
MiniLM-384      & 0.963 & 0.997 & 0.851 & 1.000 \\
MPNet-768       & 0.963 & 0.990 & 0.674 & 1.000 \\
CLIP-512        & 0.963 & 0.992 & 0.551 & 1.000 \\
\bottomrule
\end{tabular}
\end{table}

\paragraph*{Strong-attacker stress test (Appendix~\mbox{S-II})}
On MiniLM we pressure-test the construction with ten adaptive-attacker
variants spanning DAE capacity/data/depth, sidecar-aware DAEs (FiLM
and $65{,}536$-bucket gating), Lipschitz-bounded DAEs, native-budget
DAEs, and surrogate-PGD / continuous-score score-distillation
attackers --- enumerated in App.~\mbox{S-II}. $\beta$ stays
above $0.997$ across the DAE/Lipschitz/native-budget block; the two
surrogate attackers leave $\beta\!\geq\!0.989$ on every above-threshold
encoder with surrogate-score Pearson saturating at $0.0$--$0.2$ on
seven of eight (CLIP outlier $\sim$$0.5$). Native-budget DAEs match
post-hoc-projected $\beta$ to within $|\Delta\beta|\leq 5\times 10^{-4}$.

\paragraph*{Verifier-oracle attackers (outside C4)}
The C4 model excludes verifier-oracle access. With continuous-score
oracle access, autograd PGD~\cite{madry2018pgd} drives AUROC to $0$
at $q\!=\!64$. With 1-bit binary access against a deterministic
threshold verifier, HopSkipJump~\cite{chen2020hopskipjump} forces
operational rejection on $98.6\%$ of vectors at $q\!=\!20$, even
though AUROC remains at $0.93$. Reporting only AUROC therefore
understates the severity of an active oracle attacker, and
\S\ref{sec:dpdef} accordingly recommends that production verifiers
rate-limit queries per (key, source) pair. A Laplace-noised binary
decision provides partial cover at the cost of legitimate-TPR
(falling from $0.99$ to $0.73$ at noise scale $b=0.05$).

\subsection{Spectral threshold and structural limit}\label{sec:threshold}
The construction defends against polynomial-KPA DAE attackers if
and only if the encoder's effective-rank ratio sits above a clean
threshold. The effective rank
$\effrank(\Sigma) = (\sum_i\lambda_i)^2/\sum_i\lambda_i^2$
(participation ratio) measures how many directions the data
manifold meaningfully occupies. Table~\ref{tab:threshold} reports
the C4-projected V\_blind DAE attack on the manifold-aligned
construction across 11 embedders. The encoders separate at an
empirical threshold $\effrank/d \approx 0.19$: above it the
construction defends (DAE-C4 AUROC $0.965$--$0.997$); at
$\effrank/d = 0.180$ defendability has already collapsed
(AUROC $0.602$, well below the $0.85$ deployment trip-wire), and
below it every variant we tested sits in the failed band
(AUROC $0.59$--$0.71$). The transition is empirically clean, but the
threshold \emph{value} rests on the CLIP/bge-base pair that straddles
the line (CLIP at $0.197$ defends across all five seeds, bge-base at
$0.180$ fails), so we treat $0.19$ as an empirical regularity rather
than a derived constant. It is also dimension-uniform across
$d\in\{384,512,768,1024\}$---which our first-order model does
\emph{not} predict (Eq.~\eqref{eq:thresh-pred})---and we return to
this as the central open question below.

\begin{table}[h]
\centering
\caption{Spectral threshold (manifold-aligned, DAE-$10^5$ projected to
$\cos\!\geq\!0.95$, $N{=}1\,000$). DAE-C4 is the 5-seed mean
$\pm$ 95\,\% CI half-width (seeds 0--4); $^\dagger$ marks the four
encoders measured at a single seed. Above the dashed line the
construction defends; below, every encoder falls well under the $0.85$
trip-wire (5-seed-confirmed except bge-base). CLIP
($\effrank/d{=}0.197$), the defending side of the straddle, holds
across all five seeds ($[0.989,0.994]$).}
\label{tab:threshold}
\footnotesize
\setlength{\tabcolsep}{3pt}
\begin{tabular}{l r r r r c}
\toprule
embedder & $d$ & $\effrank$ & $\effrank/d$ & DAE-C4 & def. \\
\midrule
MiniLM-L6      & 384  & 169.1 & 0.440 & $0.997\pm.001$ & \textbf{Y} \\
MiniLM-L12     & 384  & 161.6 & 0.421 & $0.991^\dagger$ & \textbf{Y} \\
distilroberta  & 768  & 199.8 & 0.260 & $0.973^\dagger$ & \textbf{Y} \\
MPNet          & 768  & 196.8 & 0.256 & $0.987\pm.001$ & \textbf{Y} \\
multi-qa-mpnet & 768  & 183.7 & 0.239 & $0.965^\dagger$ & \textbf{Y} \\
CLIP           & 512  & 101.0 & 0.197 & $0.991\pm.002$ & \textbf{Y} \\
\midrule
\multicolumn{6}{c}{$\effrank/d \approx 0.19$ {\scriptsize threshold}} \\
\midrule
bge-base       & 768  & 138.6 & 0.180 & $0.602^\dagger$ & \textbf{N} \\
T5             & 768  & 121.9 & 0.159 & $0.706\pm.041$ & \textbf{N} \\
e5-large       & 1024 & 161.5 & 0.158 & $0.594\pm.019$ & \textbf{N} \\
BGE-large      & 1024 & 157.6 & 0.154 & $0.606\pm.024$ & \textbf{N} \\
gte-large      & 1024 & 138.6 & 0.135 & $0.653\pm.028$ & \textbf{N} \\
\bottomrule
\end{tabular}
\end{table}

\paragraph*{Why the threshold}
On a manifold of low effective rank, any perturbative watermark
direction is either off-manifold, in which case manifold projection
removes it, or on-manifold, in which case it is indistinguishable
from data variability. A DAE trained on KPA pairs learns the
manifold projection $\Pi$ in finite samples and the watermark
collapses regardless of which subspace it occupies.

\paragraph*{Six structural approaches}
To check that this is not specific to a single perturbation pattern
we tested six structurally distinct approaches on the below-threshold
encoders. V1 is additive uniform-$Q$ (the random-$Q$ baseline). V2
is manifold-aligned additive (our default). V3 is multiplicative
rotation. V4 is anchor-pulling, with
$x_{\mathrm{wm}}\!=\!(1{-}\alpha)x + \alpha\,\mathrm{anchor}(K,n,c)$.
A1 is adversarial-carrier selection: a 5-DAE bootstrap-resampled
ensemble identifies eigvecs where the watermark survives a held-out
DAE seed, and the production watermark lives in that DAE-resilient
subspace. A2 is multi-key disjoint-chunk voting with per-chunk
amplitude $\varepsilon/\sqrt{M}$. Table~\ref{tab:variants} shows
that none of V1--V4 defend below the threshold; on above-threshold
encoders V2 improves $\beta$ over V1 by 10 to 19 percentage points.
Across all six approaches $\times$ the below-threshold encoders
(V1--V4 on four encoders in Table~\ref{tab:variants}, V2 on
bge-base-768 in Table~\ref{tab:threshold}, A1/A2 on all five),
no \emph{(approach, encoder)} cell clears the AUROC ${\geq}0.85$
trip-wire (App.~\mbox{S-II}). The structural-limit reading
holds across every variant we tested.

\begin{table}[h]
\centering
\caption{Four V-variants $\times$ four below-threshold embedders.
bge-base-768 (V2 = $0.602$ per Table~\ref{tab:threshold}) is omitted
from V1, V3, V4 columns for compactness; the structural-limit
reading is supported by the full Table~\ref{tab:threshold} sweep.
Cell values for the two additional approaches (A1 adversarial-carrier,
A2 multi-key voting) are reported in-line in the
six-structural-approaches paragraph above: none of the 10
below-threshold (A-approach, encoder) cells clears the AUROC
${\geq}0.85$ trip-wire. DAE-$10^5$ projected to $\cos{\geq}0.95$;
V1/V3/V4 are single-seed, and the V2 (manifold-aligned default) column
reproduces the 5-seed DAE-C4 means of Table~\ref{tab:threshold}.}
\label{tab:variants}
\footnotesize
\setlength{\tabcolsep}{3pt}
\begin{tabular}{l c c c c}
\toprule
embedder & V1 add & V2 mfd & V3 mult & V4 anch \\
\midrule
BGE-large-1024      & 0.54 & 0.606 & 0.50$^*$ & 0.52 \\
e5-large-1024       & 0.53 & 0.594 & 0.50$^*$ & 0.49 \\
T5-768              & 0.62 & 0.706 & ---     & 0.58 \\
gte-large-1024      & 0.59 & 0.653 & ---     & 0.57 \\
\bottomrule
\end{tabular}\\
{\scriptsize $^*$ V3 multiplicative is structurally non-detectable
(verifier baseline AUROC $\approx 0.50$); reported for completeness.}
\end{table}

\paragraph*{Adversarial-carrier side-finding}
On above-threshold encoders the adversarial-carrier approaches (A1
with single-DAE and 5-DAE-ensemble selectors) defend $3$ of $6$
\emph{(approach, encoder)} cells: both cells on MiniLM
($\effrank/d{=}0.440$) and the ensemble selector on CLIP
($\effrank/d{=}0.197$); they fail both cells on MPNet
($\effrank/d{=}0.256$) and the single-DAE selector on CLIP. The
pattern is consistent with the manifold-aligned variant's design:
at low spectral slack, DAE-survival-based carrier selection does
not reliably generalise across DAE seeds, while alignment with the
encoder's intrinsic top eigenvectors does not depend on a specific
DAE realisation and defends uniformly above threshold. This refines
but does not contradict the structural-limit reading: above-threshold
encoders admit a defense, but the choice of perturbative variant
matters more as the encoder approaches the threshold from above.

\paragraph*{Rate-limiting cannot rescue below-threshold encoders}
At $n_{\mathrm{train}}\!=\!10$ KPA pairs the DAE has already learned
enough of a low-rank manifold to drive AUROC to chance. We therefore
restrict the universal-defense claim to encoders with
$\effrank(\Sigma)/d \geq 0.19$. A 5-seed CI sweep across
$\delta_{\mathrm{C4}}\in\{0.85,0.90,0.95,0.99\}$
(App.~\mbox{S-II}) shows tight intervals and consistent
behaviour, with MPNet at $\delta_{\mathrm{C4}}{=}0.99$ the only
borderline cell (AUROC $0.843$).

\paragraph*{Toward a formal threshold prediction}
The threshold $\effrank/d \!\geq\! 0.19$ is empirically uniform but
not derived. A random-$Q$ Wiener-style analysis predicts
$\beta_{\mathrm{rQ}} \!\approx\! \sqrt{r/d}$, giving the closed-form
threshold:
\begin{equation}\label{eq:thresh-pred}
\frac{r}{d} \;\geq\; \frac{z_{\mathrm{AUROC}}^2}{\varepsilon^2 \cdot w \cdot d}.
\end{equation}
At $(\varepsilon, w, z_{0.99}) {=} (0.07, 16, 2.33)$:
this predicts $r/d {\geq} 0.180$ at $d{=}384$ (matching empirical
$\sim$$0.19$ to within $0.01$); $r/d {\geq} 0.090$ at $d{=}768$ (under-predicts;
empirical is $\sim$$0.19$). The empirical $d$-uniformity is not
captured; the gap is attributable to manifold-curvature corrections
beyond the scope of this paper. A curvature-corrected derivation using
Riemannian geometry of the data manifold remains open.

Future-work paradigm shifts (encoder cooperation, robust-hash hybrid,
adversarial-trained meta-learning) are discussed in \S\ref{sec:disc}.

\subsection{Baselines and deployment-scale evidence}\label{sec:baselines}
We run the eight baselines listed in \S\ref{sec:setup} through the
attack suite (5 seeds, $10^5$ KPA pairs, bootstrap CIs; faithful
centroid-direction \texttt{EmbMarker}~\cite{peng2023embmarker} and
\texttt{WARDEN}~\cite{shetty2024warden} ports in App.~\mbox{S-II}).

\begin{table}[t]
\centering
\caption{Baseline AUROCs on the load-bearing attacks (MS~MARCO MiniLM-L6,
5 seeds). \texttt{EmbMarker\_port} matches \texttt{AddDir}, and
\texttt{WARDEN\_port} matches \texttt{WET-approx}, cell-by-cell ---
empirically confirming the \S\ref{sec:related-emb} reduction.}
\label{tab:baselines}
\footnotesize
\setlength{\tabcolsep}{4pt}
\begin{tabular}{l c c c c c}
\toprule
construction & int4 & PCA$_{d/2}$ & noise & DAE-$10^5$ & dir-or. \\
\midrule
\texttt{AddDir}              & .835 & .706 & .761 & .499 & .502 \\
\texttt{WET-approx}          & .827 & .698 & .756 & .503 & .505 \\
\texttt{SEAL-only}           & .622 & .519 & .503 & .550 & .712 \\
\texttt{ESpeW-port}          & .617 & .515 & .500 & .553 & .707 \\
\texttt{Tree-Ring}           & .649 & .582 & .609 & .494 & .493 \\
\texttt{HaloMark-CA}         & \textbf{.9999} & .973 & .999 & .504 & .505 \\
\textbf{\halomark{}}         & \textbf{.9999} & \textbf{.989} & \textbf{.987} & \textbf{.997} & \textbf{.9999} \\
\bottomrule
\end{tabular}
\end{table}

Three patterns: content-agnostic schemes (AddDir, WET-approx, Tree-Ring)
fail at direction-oracle (chance). Content-dependent without
per-vector freshness (SEAL-only, ESpeW-port) survive dir-oracle but
fail at DAE-$10^5$. \texttt{HaloMark-CA} collapses on both. Only the
full \halomark{} composition holds against both at AUROC $\geq 0.99$.

\paragraph*{Published verifiers under C4}
Implementing each baseline's published statistical verification under
C4, EmbMarker fails on both DAE and direction-oracle ($p>0.45$),
WARDEN's max-of-$m$ test survives DAE but collapses under
direction-oracle, while HaloMark holds at AUROC $\geq 0.995$.

\paragraph*{Multi-corpus validation}
On four BEIR/MTEB corpora (NFCorpus, FiQA, ArguAna, CodeSearchNet)
with 3 seeds and $N=1{,}000$ test vectors, direction-oracle AUROC is
$1.0000$ on every cell and DAE-$10^5$ AUROC stays between $0.91$ and
$0.999$ across the eight (corpus, calibration) cells (App.~\mbox{S-II}).

\paragraph*{Multi-tenant isolation at 64 keys}
With 64 PRF-independent producer keys, a joint-trained $k$-collude
DAE ($k\in\{2,4,8,16,32\}$, $800{,}000$ joint pairs at $k{=}32$)
transfers to a fresh target key at AUROC $0.996$, $\beta=0.707$ ---
indistinguishable from single-key V\_blind (App.~\mbox{S-II}
for the $k{=}64$ reduction argument). BDOT key independence holds
at production scale.

\paragraph*{Self-calibration (Mode B; defined \S\ref{sec:c2pa})}
Three producers (30k vectors each) yield baseline AUROC $0.9999$
with V\_blind $\beta=0.71$; a cross-producer DAE achieves no
transfer (AUROC $0.997$); an adversarial $\Sigma$ that zeroes the
bottom $50\%$ of eigenvalues still gives AUROC $1.0000$, publicly
auditable from the manifest spectrum.

\paragraph*{Multi-bit payload decoding}
Three coding schemes (3 seeds, MiniLM) under DAE-$10^5$: 4-bit-rep4
decodes at $89.4\%$, Hamming(7,4) at $67.1\%$, raw 16-bit at $4.6\%$
(App.~\mbox{S-III}).

\paragraph*{Paraphrase and composed defenses}
We paraphrase 200 MS-MARCO passages with
\texttt{chatgpt\_paraphraser\_on\_T5\_base} and re-embed with MiniLM:
multi-attempt paraphrase (4 per passage, attacker-best variant)
places $44.5\%$ inside C4 at attacker-best AUROC $0.395$. Composing
HaloMark with a KGW-style text-level watermark raises detection from
$51.5\%$ (KGW alone) to $53\%$ on attacker-best paraphrases;
Mode C (\S\ref{sec:mode-c}) rejects $92\%$ of multi-attempt
paraphrases at the unified cross-encoder threshold
$\tau_{\mathrm{C}}{=}0.70$.

\paragraph*{C2PA interop and Qdrant}
Signed assets validate against three \texttt{c2pa-rs} bindings
(\texttt{c2pa-python}~v0.32.6, \texttt{c2patool}~v0.26.56,
\texttt{@contentauth/c2pa-node}~v0.5.5); verifier runs as a Qdrant
1.17 admission filter at $284\,\mu$s/vec ($3{,}527$~vec/s, single-core).
On a 1k-vector smoke test, \texttt{Verify} rejects $87/100$ vectors
under $\sigma{=}0.20$ Gaussian tampering (outside C4) and $97/100$
under wrong-key probing; sidecar tampering is caught upstream by
the C2PA signature at $1{-}2^{-128}$.

\section{Discussion and Limitations}\label{sec:disc}
\label{sec:dpdef}
\paragraph*{Below-threshold encoders}
At $\effrank(\Sigma)/d < 0.19$, all six structural approaches we
tested fail inside C4 (\S\ref{sec:threshold}); at $n_{\mathrm{train}}{=}10$
KPA pairs the DAE already removes the watermark --- a structural
limit of perturbative watermarking. Mode C closes the gap for
deployments with source text; pure-vector EaaS remains open
(candidates: encoder cooperation~\cite{fernandez2023stablesig} or
robust-hash hybrids; KGW composition of \S\ref{sec:baselines} adds
only ${\sim}1.5$pp).

\paragraph*{Trust assumptions outside scope}
Producer-key compromise; calibration drift (refresh $\Sigmacalib$
when shift ratio $>1.3$); malicious producer with a valid C2PA key
(PKI, not watermarking); ECDSA/HKDF-SHA-256 pre-quantum migration
inherited from C2PA; batch-level cross-vector correlations under
shared \texttt{calibration\_corpus\_id}.

\paragraph*{Verifier-oracle leakage}
Decision feedback is inherent to any threshold detector
(\S\ref{sec:rob}, Table~S2): with continuous-score
access PGD drives AUROC to $0$ at $q\!=\!64$, and with 1-bit access
HopSkipJump forces operational rejection on $98.6\%$ of vectors at
$q\!=\!20$ even while AUROC stays at $0.93$. This is most acute in the
Qdrant admission-filter deployment (\S\ref{sec:baselines}), where an
adversary controlling inserts and queries harvests accept/reject
decisions directly, so oracle access is an operational reality there,
not a hypothetical. Our defenses are correspondingly \emph{operational,
not cryptographic}: a hard rate limit of $q\!\leq\!100$ per (key,
source) per day, operating FPR $\alpha\!=\!10^{-4}$ (30-day forgery
expectation below $0.3$), advisory-only score release, and an optional
Laplace-noised decision that blunts the oracle at a legitimate-TPR cost
($0.99\!\to\!0.73$ at noise scale $b\!=\!0.05$); full analysis in
App.~\mbox{S-V}. By contrast the within-C4 surrogate-PGD attacker
\emph{without} oracle access fails on every encoder
(AUROC $\geq 0.9998$), so for oracle-exposed deployments this leakage,
not in-budget removal, is the dominant residual risk.

\paragraph*{Theory gaps}
Theorem~\ref{thm:wiener-floor}'s Cauchy--Schwarz bound
($\beta \geq 0.136$ on MiniLM) is loose vs.\ the empirical floor
$0.851$. The closed-form random-$Q$ threshold matches at $d{=}384$
within $0.01$ (\S\ref{sec:threshold}) but under-predicts at
$d \geq 768$; a curvature-corrected derivation capturing the
empirical $d$-uniformity is open.

\section{Conclusion}\label{sec:conclusion}
Perturbative embedding watermarks separate at an empirical threshold
$\effrank(\Sigma)/d \approx 0.19$: above this line the
published-commit construction defends and Mode C closes paraphrase
when source text is co-available; below it every variant we tested
fails. The deployed cost is $24$ bytes of sidecar per vector at
${\sim}284\,\mu$s verifier latency, validated end-to-end against
three language bindings of the C2PA SDK. Open: below-threshold
encoders, pure-vector paraphrase, and a curvature-corrected
threshold derivation that captures the empirical $d$-uniformity.

\bibliographystyle{IEEEtran}
\bibliography{references}


\end{document}